%% file: CES_is_discontinous_in_three_user.tex
 \documentclass[conference]{IEEEtran}
\input{MATH_template.tex}
\usepackage{mathrsfs}
\onecolumn

\begin{document}

\title{ New Sufficient Conditions for Multiple-Access Channel with Correlated Sources}
\author{
 \IEEEauthorblockN{Mohsen Heidari}
  \IEEEauthorblockA{EECS Department\\University of Michigan\\ Ann Arbor,USA \\
    Email: mohsenhd@umich.edu} 

\and
 \IEEEauthorblockN{Farhad Shirani}
%  \IEEEauthorblockN{}
  \IEEEauthorblockA{EECS Department\\University of Michigan\\ Ann Arbor,USA \\
    Email: fshirani@umich.edu }

  \and
  \IEEEauthorblockN{S. Sandeep Pradhan}
  \IEEEauthorblockA{EECS Department\\University of Michigan\\ Ann Arbor,USA \\
    Email: pradhanv@umich.edu}
}
\IEEEoverridecommandlockouts

\maketitle

\begin{abstract}
The problem of three-user Multiple-Access Channel (MAC) with correlated sources is investigated. An extension to the Cover-El Gamal-Salehi (CES) scheme is introduced. We use a combination of this scheme with linear codes and propose a new coding strategy. We derive new sufficient conditions to transmit correlated sources reliably. We consider an example of three-user MAC with binary inputs. Using this example, we show strict improvements over the CES scheme.
\end{abstract}

\section{Introduction}
\IEEEPARstart{T}{he} separation principle of Shannon plays a fundamental role to reinforce the notion of modularity. This in turn allows separate development of source and channel code design. However,  as shown by Shannon \cite{Shannon}, the separation does not generalize to multi-terminal communications. For instance,  this phenomenon was observed in many-to-one communications involving transmission of correlated sources over MAC \cite{CES}. 

In the problem of MAC with correlated sources, there are multiple transmitters, each observing a source correlated to others. The transmitters do not communicate with each other and wish to send their observations via a MAC to a central receiver. The receiver reconstructs the sources losslessly. 
The separate coding approach involves a source coding part and a channel coding part.
%There are two different approaches to this problem. In the first approach the sources first compressed using a source coding scheme then a channle code is used to send the outputs of the source-encoders to via the MAC. In the other approach, a joint source-channel code is designed to map the source sequences directly to the channel's inputs. For the first case, there are two directions in the literature.  
In the channel coding part, Ahlswede \cite{Ahlswede-MAC} and Liao \cite{Liao} studied the case where the transmitters have independent information and derived the capacity region for channel coding over MAC. In the source coding part, the distributed source coding problem was studied in which transmitters can communicate to the receiver error-free.  Slepian and Wolf  showed that lossless reproduction of the sources is possible with rates close to the joint entropy \cite{Slepian-Wolf_dist}.

Due to suboptimality of the separation based strategies, the joint source-channel coding approach has been of great interest. The CES scheme introduced in \cite{CES}, is a generalization of the results in  \cite{Ahlswede-MAC} and \cite{Slepian-Wolf_MAC}. Using this scheme a single-letter characterization of the set of sources that can be reliably transmitted was derived.  It was shown that this scheme strictly improves upon the previously known strategies. However, Dueck \cite{Dueck} proved that this approach only gives a sufficient condition and not a necessary one. The joint source-channel coding problem is well studied in other settings such as: source coding with side information via a MAC \cite{Ahlswede-MAC-With-Side-info}, broadcast channels with correlated sources \cite{Han-Costa-BC} and interference channels \cite{Salehi-Kurtas-IS}.

Recently, structured codes were used to design coding strategies for joint source-channel coding problems, \cite{Pradhan-Choi, Choi-Pradhan, Nazer_Gasper_Comp_MAC, Arun_comp_over_MAC_ISIT13, transversal}. A graph-theoretic framework was introduced in  \cite{Pradhan-Choi, Choi-Pradhan} to improve the joint source-channel coding schemes both in the MAC and the broadcast channel.

In this work, we study the three-user MAC with correlated sources. We first extend the CES scheme to this problem and derive an achievable rate region. As shown in \cite{CES}, this coding strategy improves upon separate source-channel coding techniques. This is done by choosing the codewords such that they are statistically dependent on the distribution of the sources. We observe that further improvements are possible when the sources impose an algebraic structure. One example is  when one of the sources is the modulo sum of the other two. In this scenario, a structured coding strategy is needed for the codebooks to match with the structure of the sources. With this intuition, we use linear codes in the extension of the CES scheme to derive a new achievable region. Through an example, we show strict improvements over the extension of the CES scheme.

The rest of this paper is as follows: In section \ref{sec: priliminaries}, we provide the notations, definitions and the problem statement. In Section \ref{sec: CES_extension}, an extension of the CES scheme is discussed. A new coding strategy based on linear codes is provided in Section \ref{sec: proposed scheme}. The improvements over the CES scheme are discussed in Section \ref{sec: improvements over CES}. Section \ref{sec: conclusion} concludes the paper. 
\section{Priliminaries and Problem Statement}\label{sec: priliminaries}

\subsection{Notations}
In this paper, random variables are denoted using capital letters such as $X,Y$, and their realizations are shown using lower case letters such as $x,y$, respectively.  Vectors are shown using lowercase bold letters such as $\mathbf{x}, \mathbf{y}$. Sequences of number are also represented by bold letters.
 Calligraphic letters are used to denote sets such as $\mathcal{X}, \mathcal{Y}$. For any set $\mathcal{A}$, let $S_{\mathcal{A}}=\{S_a\}_{a\in \mathcal{A}}$. If $\mathcal{A}=\emptyset$, then $S_{\mathcal{A}}=\emptyset$.
 % Denote $\mathcal{A}^c=\{1,2,3\}-\mathcal{A}$.
  As a shorthand, we sometimes denote a triple $(s_1,s_2,s_3)$ by $\underline{s}$. We also denote a triple of sequences $(\mathbf{s}_1,\mathbf{s}_2,\mathbf{s}_3)$ by $\underline{\mathbf{s}}$. 
By $\FF_q$, we denote the field of integers modulo-$q$, where $q$ is a prime number.  

\subsection{Three-User MAC with Correlated Sources}
 Consider a MAC with conditional Probability Mass Function (PMF) $p(y|x_1,x_2,x_3)$, input alphabets $\mathcal{X}_j, j=1,2,3$ and output alphabet $\mathcal{Y}$. Suppose $(S_1, S_2,S_3)$ represent three sources with joint distribution $p(s_1,s_2,s_3)$.  After observing $\mathbf{S}^n_j$, the $j$th transmitter encodes it and sends the encoder's output to the channel. Upon receiving $\mathbf{Y}^n$ from the channel, the decoder wishes to reconstruct the sources losslessly. A code for this setup consists of three encoding functions $f_j: \mathcal{S}^n_j \rightarrow \mathcal{X}_j, j=1,2,3,$ and a decoding function $g: \mathcal{Y}^n \rightarrow \mathcal{S}^n_1 \times \mathcal{S}^n_2 \times \mathcal{S}^n_3.$

\begin{definition}
The source $(S_1,S_2,S_3)\sim p(s_1,s_2,s_3)$ can be \textit{reliably transmitted} over the MAC $p(y|x_1,x_2,x_3)$, if for any $\epsilon >0$, there exist  encoding functions $f_1,f_2,f_3$ and a decoding function $g$ such that 
$$P\{g(Y^n) \neq (S_1^n,S_2^n, S_3^n| X^n_i=f_i(S_i^n), i=1,2,3\}\leq \epsilon.$$
\end{definition}

\subsection{Common part}
To define the common parts between the sources, we use the notion given in \cite{Witsenhausen}.
\begin{definition}\label{def: common part}
Consider random variables $S_j, j=1,2, \dots, m$. $W$ is defined as the \textit{common part} among these random variables by finding maximum positive integer $k$ for which there exist functions $$f_j: \mathcal{S}_j \rightarrow \{1,2,\dots, k\}, \quad j=1,2,\dots , m$$ with $P\{f_j(S_j)=i\}>0$ for all $ i\in \{1,2,\dots, k\}$ and $j\in \{1,2,\dots, m\}$ such that $W=f_j(S_j), j=1,2,.., m $, with probability one. 
\end{definition}

\section{A Three-User Extension of the  CES Scheme}\label{sec: CES_extension}
In this section, we first review the CES scheme for the problem of two-user MAC with correlated sources. Then we introduce an extension to the scheme in the three-user case.
Consider a MAC with conditional PMF $p(y|x_1,x_2)$. Let $S_1$ and $S_2$ be two correlated sources and $W$ be the common part between them as defined in Definition \ref{def: common part}.  In the CES scheme, first the common part $W$ is calculated at each encoder. Since both encoders have access to $W$, they can fully cooperate to encode it (as if it is done by a centralized encoder).  Next at each transmitter, each source is encoded using a codebook that is ``super-imposed" on the common codebook. 

It is shown in \cite{CES} that using this scheme, reliable transmission of $S_1$ and $S_2$ is possible if the following holds
\begin{align*}
H(S_1|S_2) & \leq I(X_1;Y|X_2, S_2, U),\\
H(S_2|S_1) & \leq I(X_2;Y|X_1, S_1, U),\\
H(S_1,S_2|W) & \leq I(X_1X_2;Y|W, U),\\
H(S_1, S_2)  & \leq I(X_1 X_2;Y),
\end{align*} 
 where  %\begin{align*}
$p(s_1,s_2,u,x_1,x_2, y)=p(s_1,s_2) p(u)p(x_1|s_1,u)p(x_2|s_2,u)p(y|x_1,x_2).$
%\end{align*} 

We use the above argument to extend the CES scheme for sending correlated sources over a three-user MAC. Consider the sources $S_1,S_2,S_3$. We use Definition \ref{def: common part} to construct four different common parts among the sources. Let $W_{ij}$ be the common part of $S_i, S_j$. For more convenience, we denote the common part of $S_i$ and $S_j$, either by $W_{ij}$ or $W_{ji}$ (we simply drop the condition $j>i$, as it is understood that $W_{ij}=W_{ji}$).  Lastly, $W_{123}$ is the common part of $S_1,S_2$ and $S_3$.  

By observing $S_i$ at the $i$th transmitter, three common parts can be calculated, $W_{123}$ and $W_{ij}, j\neq i$. The three-user extension of CES involves three layers of coding. In the first layer $W_{123}$ is encoded at each encoders. Next, based on the output of the first layer, the $W_{ij}$'s are encoded.  Finally, based on the output of the first and the second layers, $S_1,S_2$ and $S_3$ are encoded.

The following preposition determines sufficient conditions for which correlated sources can be transmitted using this scheme. 

\begin{preposition}\label{prep: CES_three_user}
The source $(S_1,S_2,S_3)\sim p(s_1,s_2,s_3)$ can be reliably transmitted over a MAC with conditional probability $p(y|x_1x_2x_3)$, if for distinct $i,j,k\in \{1,2,3\}$ and any $\mathcal{B}\subseteq \{12,13,23\}$ the following holds:
\begin{align*}
H(S_i|S_jS_k)&\leq I(X_i;Y|S_jS_kX_jX_k U_{123} U_{12}U_{13} U_{23})\\
H(S_iS_j|S_kW_\mathcal{B})&\leq I(X_iX_j;Y|S_k W_\mathcal{B} U_{123}U_{ik}U_{jk}U_\mathcal{B}X_k)\\
%H(S_iS_j|S_kW_{ij})&\leq I(X_iX_j;Y|S_k W_{ij} X_k  U_{123}U_{ik}U_{jk} U_{ij})\\
%H(S_1|S_2S_3)&\leq I(X_1;Y|S_2S_3X_2X_3 U_{(\mathbf{B}\subseteq \{1,2,3\})})\\
%H(S_{\mathcal{A}} | S_{\mathcal{A}^c}) &\leq I(X_{\mathcal{A}};Y| X_{\mathcal{A}^c}  S_{\mathcal{A}^c}U_{(\mathcal{B}\supseteq \mathcal{A}^c)})\\
%H(S_{\mathcal{A}} | S_{\mathcal{A}^c} W_{\mathcal{A}}) &\leq I(X_{\mathcal{A}};Y| X_{\mathcal{A}^c}  S_{\mathcal{A}^c} W_{\mathcal{A}}%U_{(\mathcal{B}\supseteq \mathcal{A}^c)} U_{\mathcal{A}})\\
%H(S_1S_2S_3|W_{123})&\leq I(X_1X_2X_3;Y|W_{123} U_{123})\\
H(S_1S_2S_3|W_{123} W_{\mathcal{B}})&\leq I(X_1X_2X_3;Y|W_{123}W_{\mathcal{B}}U_{123}U_{\mathcal{B}})\\
H(S_1S_2S_3)&\leq I(X_1X_2X_3;Y),
\end{align*}
where $U_{ij}=U_{ji}$ and
\begin{align}\label{eq: p_sux}
p(\underline{s},& \underline{x}, u_{123},u_{12},u_{13},u_{23})=p(\underline{s})p(u_{123})[\prod_{b\in \{12,13,23\}}p(u_b|w_{b}u_{123})] \cdot [\prod_{\substack{i, j, k \in \{1,2,3\}\\ j<k}} p(x_i|s_iu_{123}u_{ij}u_{ik})].
\end{align}
\end{preposition}

\begin{proof}[Outline of the proof]
Let the random variables $X_1,X_2,X_3$, $U_{123}, U_{12},U_{13}$ and $U_{23}$ be distributed according to the above theorem. Let the $n$-length sequence $\mathbf{s}_i$ be a realization of the source $S_i$, where $i=1,2,3$.
\paragraph*{\textbf{Codebook Generation}}
For each $\mathbf{w}_{123}\in \mathcal{W}_{123}$ randomly generate a sequence $\mathbf{u}_{123}$ according to the PMF of $U_{123}$. Index them by $u_{123}(\mathbf{w}_{123})$. For each $\mathbf{u}_{123}$ and $\mathbf{w}_{b}, b \in\{12,13,23\}$ randomly generate  a sequence according to $p(u_{b}|w_{b}u_{123})$. Index them by $u_{b}(\mathbf{w}_{b}, \mathbf{u}_{123})$.

For each $\mathbf{s}_i$, first find the corresponding sequences of the common parts $\mathbf{w}_{123}, \mathbf{w}_{ij}$ and $\mathbf{w}_{ik}$, where $i,j,k \in \{1,2,3\}$ are distinct. Next find the corresponding sequences  $u_{123}(\mathbf{w}_{123}), u_{ij}(\mathbf{w}_{ij}, \mathbf{u}_{123})$ and $u_{ik}(\mathbf{w}_{ik}, \mathbf{u}_{123})$ as generated above. Lastly generate a sequence $\mathbf{x}_i$ randomly and independently according to $p(x_i|s_iu_{123}u_{ij}u_{ik})$. For shorthand we denote such sequence by $x_i(\mathbf{s}_i, \mathbf{u}_{123},\mathbf{u}_{ij},\mathbf{u}_{ik})$.

\paragraph*{\textbf{Encoding}}
Upon observing the output $\mathbf{s}_i$ of the source, the $i$th transmitter first calculates the common part sequences $\mathbf{w}_{123}, \mathbf{w}_{ij}$ and $\mathbf{w}_{ik}$. Then at the first stage it finds $u_{123}(\mathbf{w}_{123})$. At the second stage, it finds $u_{ij}(\mathbf{w}_{ij}, \mathbf{u}_{123})$ and $u_{ik}(\mathbf{w}_{ik}, \mathbf{u}_{123})$. Lastly, at the third stage, it sends $x_i(\mathbf{s}_i, \mathbf{u}_{123},\mathbf{u}_{ij},\mathbf{u}_{ik})$.

\paragraph*{\textbf{Decoding}}
Upon receiving $\mathbf{y}$ from the channel, the decoder finds $\underline{\tilde{\mathbf{s}}}=(\tilde{\mathbf{s}}_1,\tilde{\mathbf{s}}_2,\tilde{\mathbf{s}}_3)$ such that 
\begin{align*}
(\underline{\tilde{\mathbf{s}}}, \tilde{\mathbf{u}}_{123}, \tilde{\mathbf{u}}_{12}, \tilde{\mathbf{u}}_{13}, \tilde{\mathbf{u}}_{23}, \tilde{\mathbf{x}}_1,\tilde{\mathbf{x}}_2,\tilde{\mathbf{x}}_3,\mathbf{y})\in A_{\epsilon}^{(n)}(\underline{S},U_{123},U_{12},U_{13},U_{23},X_1,X_2,X_3,Y)
\end{align*}

where $ \tilde{\mathbf{u}}_{123}=u_{123}(\tilde{\mathbf{w}}_{123}),\tilde{\mathbf{u}}_{ij}=u_{ij}(\tilde{\mathbf{w}}_{ij}, \tilde{\mathbf{u}}_{123})$. Note $\tilde{\mathbf{w}}_{123}, \tilde{\mathbf{w}}_{ij}$ are the corresponding common part sequences of $\tilde{\mathbf{s}}_1,\tilde{\mathbf{s}}_2,\tilde{\mathbf{s}}_3$.

A decoding error will be occurred, if no unique $(\tilde{\mathbf{s}}_1,\tilde{\mathbf{s}}_2,\tilde{\mathbf{s}}_3)$ is found. Using a standard argument as in \cite{CES}, it can be shown that the probability of error approaches zero, if the conditions in Preposition \ref{prep: CES_three_user} are satisfied.

\end{proof}

\section{New Sufficient Conditions}\label{sec: proposed scheme}
In this section, new sufficient conditions for transmission of correlated sources are derived. We enhance the previous scheme using linear codes. 

%It is shown in \cite{CES} that CES improves upon separate source channel coding techniques. This is due to the fact the codewords are statistically depend on the distribution of the sources. which leads to correlation preserving mapps. In other words,This is in contrast with seperate source channel coding, where the codewords are generated independent of the source distribution. Further improvements is possible by exploiting the common part of the sources. Since the common part is available at all of the corresponding encoders, it can be compressed as if it is done by a centralized transmitter.
% 
%The common parts used in the CES scheme are based on univariate functions as in Definition \ref{def: common part}. We extend this notion of common part to bivariate functions as in [Arun-IC]. 

The common parts used in CES are defined using univariate functions as in Definition \ref{def: common part}. In the case of more than two sources, the  notion of common parts can be extended using bivariate functions as in \cite{Arun-IC}. The following example provides a triple of sources with a bivariate common part.

 \begin{example}\label{ex: bivariate}
Let $S_1,S_2,S_3$ be three binary sources, where $S_1$ is independent of $S_2$ and $S_3=S_1 \oplus_2 S_2$, with probability one. Here, $S_3$ is a bivariate common part of $S_1$ and $S_2$. However, there is no univariate common parts among the sources.
 \end{example}
 
%In general, bivariate common parts are not exploited in unstructured coding schemes; because each encoder has access only to one source and only univariate functions can be calculated. However, for the sources in Example \ref{ex: bivariate}, linear codes can be used to exploit the bivariate common parts.

Next, we present a linear coding scheme for the above example. Select a linear code with a generator matrix $\mathbf{G}$ chosen randomly and uniformly on $\FF_2$.  The $i$th transmitter encodes $S^n_i$  using this linear code. Since $S_3=S_1\oplus_2 S_2$, using this approach $X_3=X_1\oplus_2 X_2$, with probability one. In this case, $X_1$ and $X_2$ are independent and uniform. In contrast, using randomly generated unstructured codes as in the extension of CES, the equality can not hold (unless the encoders are trivial). More precisely, in the CES scheme, given $S_1,S_2,S_3$ the random variables $X_1,X_2,X_3$ are mutually independent. Hence, one can conclude that for all valid joint distributions for CES, with high probability, $X_3 \neq X_1\oplus_2X_2$ . This assertion is discussed in more detail in the proof of Lemma \ref{lem: neigh of gamma* is not achievable by CES}.

We use the intuition behind the argument above and propose a new coding strategy in which a combination of linear codes and the CES scheme is used. We define a new class of common parts. This class of common parts are linked with our understanding of bivariate common information [18]. The new common part is called a \textit{$q$-additive common part}. The common part consists of a vector of random variables $(T_1,T_2,T_3)$. Here, $T_i$ is available at the $i$th encoder. In contrast to the univariate common parts in the CES scheme, these three random variables are not equal. Rather, each of them is a linear combination of the other two. This linear structure can be exploited using structured codes. The next definition formalizes this notion. 

\begin{definition}
For a prime number $q$, we say that $(T_1,T_2, T_3)$ is a \textit{$q$-additive common part} of $(S_1,S_2,S_3)$, if there exist functions $f_1, f_2,f_3$ such that  with probability one 1) $T_i=f_i(S_i)$, 2) $T_3=T_1 \oplus_q T_2$, 3) $T_i$ are nontrivial random variables.  
\end{definition}

%This scheme involves two stages of coding. First, we use linear codes to generate $V^n_i=T_i^n\mathbf{G}$. Then,  conditioned on $V_i$'s, we build an unstructured random code to encode the sources.

In the following Theorem, we derive sufficient conditions for transmission of correlated sources. We will show in Section \ref{sec: improvements over CES} that this leads to enlarging the class of correlated sources that can be reliably transmitted.

\begin{theorem}\label{them: achievable-rate-for -proposed scheme}
The source $(S_1,S_2,S_3)$ can be reliably transmitted over a MAC with conditional PMF $p(y|x_1,x_2,x_3)$ if for any distinct $i,j,k \in \{1,2,3\}$ and for any $\mathcal{A}\subseteq \{1,2,3\}, \mathcal{B}\subseteq \{12,13, 23\}$ the followings hold:
\begin{align}\label{eq: transmittable bound 1}
%H(S_{\mathcal{A}} | S_{\mathcal{A}^c}) &\leq I(X_{\mathcal{A}};Y| X_{\mathcal{A}^c} V_{\mathcal{A}^c} S_{\mathcal{A}^c}),\\
%H(S_1S_2S_3 |T_{1}T_{2}T_{3}) &\leq I(X_{\mathcal{A}};Y|  T_{\mathcal{A}^c}V_{\mathcal{A}^c} )
%
H(S_i|S_jS_k)&\leq I(X_i;Y|S_j S_kU_{123}U_{12}U_{13}U_{23}V_1V_2V_3 X_jX_k)\\
%H(S_iS_j|S_k)&\leq I(X_iX_j;Y|S_k, U_{123}U_{ik}U_{jk}, V_k, X_k)\\
%H(S_iS_j|S_kT_i)&\leq I(X_iX_j;Y|S_k, U_{123}U_{ik}U_{jk}, T_i~ V_{\{1,2,3\}},  X_k)\\
%H(S_iS_j|S_k W_{ij})&\leq I(X_iX_j;Y|S_k W_{ij}, U_{123}U_{12}U_{13}U_{23}, V_k, X_k)\\
%H(S_iS_j|S_k W_{ij}T_i)&\leq I(X_iX_j;Y|S_k W_{ij}, U_{123}U_{12}U_{13}U_{23}, T_i ~V_{\{1,2,3\}}, X_k)\\
H(S_iS_j|S_k W_{\mathcal{B}}T_\mathcal{A})&\leq I(X_iX_j;Y|S_k W_{\mathcal{B}} U_{123}U_{ik}U_{jk}U_{\mathcal{B}} T_\mathcal{A} V_{k}V_\mathcal{A}X_k)\\
H(S_iS_jS_k|W_{123}W_\mathcal{B}T_\mathcal{A})&\leq I(X_iX_jX_k;Y|W_{123}W_\mathcal{B}U_{123}U_\mathcal{B}T_\mathcal{A}V_\mathcal{A})\\
H(S_iS_jS_k|T_\mathcal{A})&\leq I(X_iX_jX_k;Y|T_\mathcal{A} V_\mathcal{A})\label{eq: transmittable bound end}
%H(S_iS_jS_k|W_{123})&\leq I(X_iX_jX_k;Y|W_{123} U_{123})\\
%H(S_iS_jS_k) & \leq I(X_iX_jX_k;Y)\label{eq: transmittable bound end}
\end{align}
where 1)~$(T_1,T_2,T_3)$ is a $q$-additive common part of the sources for a prime $q$, 2) $p(u_{123}u_{12}u_{13}u_{23}|\underline{s})$ is the same as in (\ref{eq: p_sux}), 3) the Markov chain $U_{123}U_{\{12,13,23\}} \leftrightarrow S_1S_2S_3\leftrightarrow V_1V_2V_3$ holds, 4) $V_3=V_1\oplus_q V_2$ with probability one and $p(v_1,v_2)=\frac{1}{q^2}$, 
$$
5) ~p(\underline{x}|u_{123}u_{12}u_{13}u_{23},\underline{v},\underline{s})= \prod_{\substack{i, j, k \in \{1,2,3\}\\ j<k}} p(x_i|s_iu_{123}u_{ij}u_{ik}v_i).
$$
\end{theorem}

%\begin{remark}
%Any source $(S_1,S_2,S_3)$ with no univariate common parts that satisfies the bounds in Preposition \ref{prep: CES_three_user} also satisfies (\ref{eq: bounds for S*}).
%\end{remark}

\begin{remark}
Suppose the source $(S_1,S_2,S_3)$ has no univariate common part. Consider the set of such sources that can be transmitted either using linear codes or using the CES scheme.  This set is included in the set of sources that satisfy (\ref{eq: transmittable bound 1})-(\ref{eq: transmittable bound end}).
\end{remark}
%The inclusion stated in the Remark above is straightforward to prove. We show the strictness through an example in Section \ref{sec: improvements over CES}.
% 
%
%
%\begin{remark}
%In \cite{complete-version}, by considering univariate common parts, we extend the results of Theorem \ref{them: achievable-rate-for -proposed scheme}.
%\end{remark}

\begin{proof}[Outline of the proof]
Consider only the special case where $S_3=S_1\oplus_q S_2$ and there is no univariate common part. In this situation, set $T_i=S_i$ and fix probability mass function $p(x_i|s_i v_i)$, where $i=1,2,3$. Generate $\mathbf{b}_1$, $\mathbf{b}_2 \in \FF_q^n$ and an $n \times n$ matrix $\mathbf{G}$ with elements selected randomly, uniformly and independently from $\FF_q$. Set $\mathbf{b}_3=\mathbf{b}_1\oplus_q \textbf{b}_2$. 
\paragraph*{\textbf{Codebook Generation}}
  For each sequence $\mathbf{s}_i$, define $v_i(\mathbf{s}_i)=\mathbf{s}_i \mathbf{G}\oplus_q \mathbf{b}_i,$ where all the additions and multiplications are modulo-$q$. For each sequence $\textbf{s}_i, \mathbf{v}_i \in \FF_q^n$ independently generate $\textbf{x}_i$ according to $\prod_{j=1}^n p(x_{i,j}|s_{i,j} v_{i,j})$. Index them by $x_i(\mathbf{s}_i, \mathbf{v}_i)$. % $V_{1j}$ and $V_{2j}$ are uniformly distributed over$\{0,1\}$ and are independent of each other and $S_1$ and $S_2$. If $S_3 =S_1 \oplus_2 S_2$, then $V_{3j}=V_{1j} \oplus_2 V_{2j}$. Otherwise $V_{3j}$ is uniformly distributed and independent of $V_{1j} $ and $ V_{2j}$.

\paragraph*{\textbf{Encoding}}
Given the sequence  $\mathbf{s}_i$, encoder $i$ first finds  $v_i(\mathbf{s}_i)$, then sends $x_i(\mathbf{s}_i, v_i(\mathbf{s}_i))$. 

\paragraph*{\textbf{Decoding}}
Upon receiving $\mathbf{y}$ from the channel, the decoder finds $\tilde{\mathbf{s}}_1, \tilde{\mathbf{s}}_2$ and $\tilde{\mathbf{s}}_3$ such that 
$
(\underline{\mathbf{\tilde{s}}},\mathbf{v}_1(\tilde{\mathbf{s}}_1),\mathbf{v}_2(\tilde{\mathbf{s}}_2),\mathbf{v}_3(\tilde{\mathbf{s}}_3), \tilde{\mathbf{x}}_1, \tilde{\mathbf{x}}_2, \tilde{\mathbf{x}}_3, \mathbf{y})\in A_{\epsilon}^{(n)},
$
where $\tilde{\mathbf{x}}_i=x_i(\tilde{\mathbf{s}}_i,\mathbf{v}_i(\tilde{\mathbf{s}}_i))$. 
An error is declared, if no unique $(\tilde{\mathbf{s}}_1,\tilde{\mathbf{s}}_2,\tilde{\mathbf{s}}_3)$ were found. 

We show in Appendix \ref{sec: proof of them lin_ces} that the probability of error approaches zero as $n\rightarrow \infty$, if (\ref{eq: transmittable bound 1})-(\ref{eq: transmittable bound end}) are satisfied. For a general $(S_1,S_2,S_3)$ the proof follows by the above argument and the proof of Preposition \ref{prep: CES_three_user}.

%The decoding is performed in two stages. At the first stage,  $\mathbf{u}_i$ are treated as noise and the decoder tries to reconstruct  $\mathbf{s}_1 \oplus \mathbf{s}_2$. Since $\mathbf{v}_i$ are independent of $\mathbf{u}_i$, there is an equivalent MAC from $V_i$ to $Y$. Upon receiving $\mathbf{y}$, the decoder finds sequences $(\mathbf{\tilde{s}}_1,\mathbf{\tilde{s}}_2)\in  A_{\epsilon}^{(n)}(S_1,S_2)$ such that
%\begin{equation}
%(\mathbf{\tilde{s}}_1,\mathbf{\tilde{s}}_2, \mathbf{v}_1(\mathbf{\tilde{s}}_1),\mathbf{v}_2(\mathbf{\tilde{s}}_2),\mathbf{v}_3(\mathbf{\tilde{s}}_1 \oplus_2\mathbf{\tilde{s}}_2 ), \mathbf{y})\in A_{\epsilon}.
%\end{equation}
%The error event $E_{sum}$ is declared if $\mathbf{\tilde{s}}_1 \oplus_2 \mathbf{\tilde{s}}_2$ is not unique.
%
%Suppose there is no error at the first stage. Then, $\mathbf{s}_1 \oplus_2 \mathbf{s}_2$ is reconstructed losslessly and  $\mathbf{v}_1 \oplus \mathbf{v}_2$ is known. 
% that the decoder the linear coding part of the scheme is used to decode $\mathbf{s}_1 \oplus \mathbf{s}_2$. Then given such information on the $\mathbf{s}_1$ and $\mathbf{s}_2$ are decoded.Upon receiving $\mathbf{y}$ from the channel, the decoder first tries to reconstruct $\mathbf{s}_1 \oplus_2 \mathbf{s}_2$ using the linear part of the code.    

\end{proof}

\section{Improvements over the CES Scheme} \label{sec: improvements over CES}
In this section, through an example, we show that Theorem \ref{them: achievable-rate-for -proposed scheme} strictly enlarges the class of correlated sources that can be transmitted reliably using linear codes or the CES scheme. We introduce a setup consisting of a source triple and a MAC. 
\begin{example}\label{ex: CES is suboptimal}
Consider binary sources $S_1, S_2, S_3$, where  $S_1$ and $S_3$ are independent and $S_3=S_1\oplus_2 S_2$. Let $S_1 \sim Be(\sigma)$ and $ S_3 \sim Be(\gamma)$, where $\sigma, \gamma \in [0,1]$.  Consider the MAC in Figure \ref{fig: Exp1}, where the input alphabets are binary.  $N$ is independent of other random variables and is distributed according to Table \ref{tab: N}, where $0 \leq \delta \leq \frac{1}{2}, \delta \neq \frac{1}{4}$.
\begin{table}[h]
\caption {Distribution of $N$}\label{tab: N}
\begin{center}
\begin{tabular}{c|c|c|c|c}
N & 0 & 1 & 2 & 3\\
\hline
$P_N$ & $\frac{1}{2}-\delta$ & $\frac{1}{2}$ & $\delta$ & $0$
\end{tabular}
\end{center}
\end{table}

\begin{figure}[h]
\centering
\includegraphics[scale=0.8]{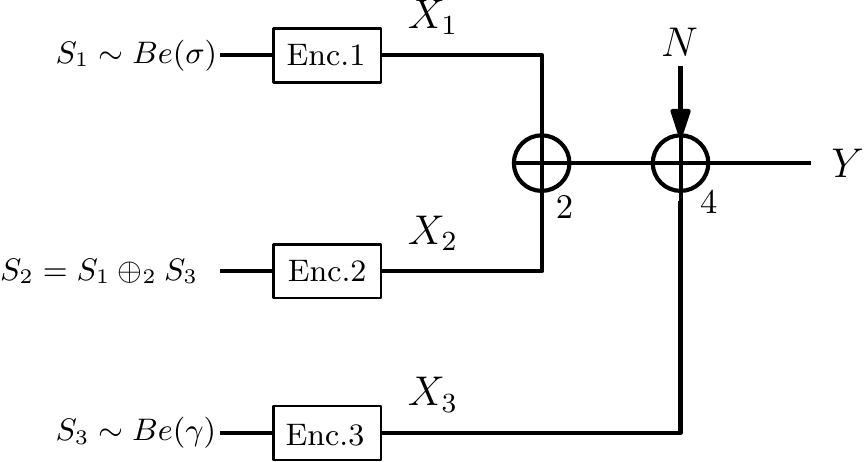}
\caption{The diagram the setup introduced in Example \ref{ex: CES is suboptimal}. Note the input alphabets of this MAC are restricted to $\{0,1\}$.}
\label{fig: Exp1}
\end{figure}
\end{example}

For this setup, we show that there exist a $\sigma$ and $\gamma$ whose corresponding sources in Example \ref{ex: CES is suboptimal} cannot be transmitted reliably using the CES scheme. However,  based on Theorem \ref{them: achievable-rate-for -proposed scheme}, this sources can be reliably transmitted.

\begin{remark}\label{rem: gamma}
Let $\sigma=0$. In this case, $S_1=0$ and $S_2=S_3$,  with probability one. From Proposition \ref{prep: CES_three_user}, $(S_1,S_2,S_3)$ can be transmitted using the CES strategy, as long as $h(\gamma)\leq 2-H(N)$. 
\end{remark}

We find a $\gamma$, in Remark \ref{rem: gamma}, such that $h(\gamma)=2-H(N)$. Since $2-H(N)\leq 1$,  we can calculate  $h^{-1}(2-H(N))$. This gives two candidates for $\gamma$. We select the one that is less than $1/2$ and denote it by $\gamma^*$.

 By Remark \ref{rem: gamma}, the source $(S_1,S_2,S_3)$ with $\sigma=0$ and  $\gamma=\gamma^*$ can be transmitted using the CES scheme. However, we argue that for small enough $\epsilon$ the source $(S_1,S_2,S_3)$ with $\sigma=\epsilon$ and  $\gamma=\gamma^*-\epsilon$ cannot be transmitted using this scheme (Lemma \ref{lem: neigh of gamma* is not achievable by CES}). Whereas, from Theorem \ref{them: achievable-rate-for -proposed scheme}, this source can be transmitted reliably (Lemma \ref{lem: neigh of gamma* is achievable}).

\begin{lem}\label{lem: neigh of gamma* is not achievable by CES}
Consider the setup in Example \ref{ex: CES is suboptimal}. $\exists ~\epsilon >0$ such that for any  $\sigma >0$ and $\gamma\geq \gamma^* - \epsilon$, the source $(S_1,S_2,S_3)$ corresponding to $\sigma$ and $\gamma$ cannot be transmitted using the three-user CES strategy. 
\end{lem}

\begin{lem}\label{lem: neigh of gamma* is achievable}
 $\exists ~ \epsilon' >0$ such that for any  $ \sigma \leq \epsilon'$ and $|\gamma-\gamma^*|\leq \epsilon'$, the source $(S_1,S_2,S_3)$ corresponding to $\sigma$ and $\gamma$, as in Example \ref{ex: CES is suboptimal}, can be transmitted.
\end{lem}
The proof of Lemma \ref{lem: neigh of gamma* is not achievable by CES} and \ref{lem: neigh of gamma* is achievable} is given in Appendix \ref{sec: proof_neigh is not achievable by CES} and \ref{sec: proof_ neigh is achievable}, respectively.
\begin{remark}
Consider $\epsilon$ and $\epsilon'$ as in Lemma \ref{lem: neigh of gamma* is not achievable by CES} and \ref{lem: neigh of gamma* is achievable}, respectively. Take $\epsilon''=\min\{\epsilon, \epsilon'\}$. As a result of these lemmas, the source $(S_1,S_2,S_3)$ corresponding to $\sigma=\epsilon''$ and $\gamma=\gamma^*-\epsilon''$ can be transmitted reliably while it cannot be transmitted using the CES scheme.
\end{remark}

\section{Conclusion}\label{sec: conclusion}
Transmission of correlated sources over thee-user MAC was investigated in this paper and an extension of the CES strategy was presented for this problem. We characterized sufficient conditions for which reliable transmission of correlated sources is possible using this scheme. Then by proposing a new coding technique, we enlarged the set of sources that can be transmitted reliably.

\appendices
\section{Proof of Lemma \ref{lem: neigh of gamma* is not achievable by CES}}\label{sec: proof_neigh is not achievable by CES}
\begin{proof}
We first derive an outer bound for the CES scheme. Consider the fourth inequality in Preposition  \ref{prep: CES_three_user}.  Since $\sigma>0$ there is no common part. Let $U'=U_{123}U_{12}U_{13}U_{23}$. Suppose the source $(S_1,S_2,S_3)$ in Example \ref{ex: CES is suboptimal}  can be transmitted using the CES, then  the following holds
\begin{equation}\label{eq: CES outer bound 1}
h(\gamma)+h(\sigma) \leq \max_{p(u')p(\underline{x}|u'\underline{s})} I(X_1X_2 X_3;Y|U'),
\end{equation}
where $$p(\underline{s},\underline{x}, u')=p(\underline{s})p(u')p(x_1|s_1, u')p(x_2|s_2,u')p(x_3|s_3,u').$$
It can be shown that the right-hand side in (\ref{eq: CES outer bound 1}) is equivalent to 

\begin{equation}\label{eq: CES outter bound 2}
h(\gamma)+h(\sigma) \leq \max_{p(\underline{x}|\underline{s})} I(X_1X_2 X_3;Y),
\end{equation}
where $p(\underline{s},\underline{x})=p(\underline{s})p(x_1|s_1)p(x_2|s_2)p(x_3|s_3).$

Next, we argue that the right-hand side in (\ref{eq: CES outter bound 2}) is strictly less than $h(\gamma^*)=2-H(N)$. For the moment assume this argument is true. Then by the bound above, $h(\gamma)+h(\sigma)< h(\gamma^*)$. This implies that $\exists \epsilon_0>0$ such that for any $\sigma$,  $h(\gamma^*)-h(\gamma)>\epsilon_0$. Hence, as the entropy function is continuous, $\exists \epsilon>0$ such that any source with $\sigma>0$ and $\gamma\geq\gamma^*-\epsilon$ cannot be transmitted using the CES scheme.

It remains to show that the right-hand side in (\ref{eq: CES outter bound 2}) is strictly less than $2-H(N)$. Note  $I(X_1,X_2,X_3;Y)=H(Y)-H(N)$. Hence, we need to show $H(Y)<2$. We proceed by finding all the necessary and sufficient conditions on $p(x_1, x_2,x_3)$ for which $Y$ is uniform over $\ZZ_4$. Then we show that since the distributions taken for maximization in (\ref{eq: CES outter bound 2}) do not satisfy these conditions.

From Figure \ref{fig: Exp1}, $Y= (X_1\oplus_2 X_2) \oplus_4 X_3 \oplus_4 N$. Denote $X'_2=X_1\oplus_2 X_2$. Let $P(X'_2 \oplus_4 X_3=i)=q(i)$ where $i=1,2,3,4$. Since $X'_2$ and $X_3$ are binary, $q(3)=0$. Given the distribution of $N$ is Table \ref{tab: N},  the distribution of $Y$ is as follows:
\begin{align*}
P(Y=0)&=q(0)(\frac{1}{2}-\delta)+q(2)\delta, \quad\quad
P(Y=1)=q(0)\frac{1}{2}+q(1)(\frac{1}{2}-\delta)\\
P(Y=2)&=q(0)\delta+q(2) (\frac{1}{2}-\delta),\quad\quad
P(Y=3)=q(2)\frac{1}{2}+q(1)\delta
\end{align*}

Assume  $\delta \neq \frac{1}{4}$. By comparing the first and third bounds, we can show that $Y$ is uniform, if and only if $q(1)=0$ and $q(0)=q(2)=\frac{1}{2}$. Note
 \begin{align*}
 q(1)=P(X'_2=0,X_3=1)+P(X'_2=1,X_3=0)
 \end{align*}
 Therefore, $q(1)=0$ implies that  $X_3=X'_2$ with probability one. If this condition is satisfied, then $q(0)=P(X_3=0)$ and $q(2)=P(X_3=1)$. Since $q(0)=q(2)=\frac{1}{2}$ then $X_3$ is uniform over $\{0,1\}$.
To sum up, we proved that $Y$ is uniform, if and only if 1) $X_3=X_1\oplus_2 X_2$. 2) $X_3$ is uniform over $\{0,1\}$.

Note the distributions given in the CES scheme for this case satisfy the Markov chain $X_3-S_3-X_1,X_2$. Hence, we can show for these distributions, the condition  $X_3=X_1\oplus_2 X_2$ hold if and only if $X_3$ is a function of $S_3$. However, as $\gamma<1/2$, $X_3$ cannot be uniform over $\{0,1\}$. This contradicts with the second condition.
\end{proof}

\section{Proof of Lemma \ref{lem: neigh of gamma* is achievable}}\label{sec: proof_ neigh is achievable}
\begin{proof}
 For the setup in Example \ref{ex: CES is suboptimal}, the bounds given in Theorem \ref{them: achievable-rate-for -proposed scheme} are simplified to 
\begin{align}\label{eq: transmittable sources bound 1}
h(\gamma) &\leq I(X_2 X_3;Y| X_1 S_1V_1)\\\label{eq: transmittable sources bound 2}
h(\sigma) &\leq I(X_1 X_2; Y|X_3 S_3 V_3)\\\label{eq: transmittable sources bound 3}
h(\gamma)+h(\sigma)-h(\sigma * \gamma)&\leq I(X_1 X_3;Y| X_2 S_2 V_2)\\\label{eq: transmittable sources bound 4}
h(\gamma)+h(\sigma)&\leq I(X_1 X_2 X_3;Y).
\end{align}

Set $X_i=V_i, i=1,2,3$, where the distribution of these random variables are given in Theorem \ref{them: achievable-rate-for -proposed scheme}. One can verify that the source corresponding to $\sigma=0$ and $\gamma=\gamma^*$ satisfies the above inequalities and therefore can be transmitted. 

No that all the terms in (\ref{eq: transmittable sources bound 1})-(\ref{eq: transmittable sources bound 4}) are entropy functions and mutual information. Therefore, they are continuous with respect to conditional density $p(\underline{x}|\underline{s},\underline{v})$. Hence, one can show that $\forall \epsilon_0>0$, there exist a conditional density $p(\underline{x}|\underline{s},\underline{v})$  such that 
\begin{align*}
I(X_2 X_3;Y| X_1 S_1V_1)& \geq 2-H(N)-\eta(\epsilon_0)\\
I(X_1 X_2; Y|X_3 S_3 V_3)&\geq \epsilon_0\\
I(X_1 X_3;Y| X_2 S_2 V_2)&\geq 2-H(N)-\eta(\epsilon_0)\\
I(X_1 X_2 X_3;Y) &\geq 2-H(N)-\eta(\epsilon_0),
\end{align*}
where $\eta()$ is function of $\epsilon$ such that $\eta(\epsilon)\rightarrow 0$ as $\epsilon \rightarrow 0$.

Note also that the left-hand sides in (\ref{eq: transmittable sources bound 1})-(\ref{eq: transmittable sources bound 4}) are continuous in $\sigma$ and $\gamma$. Hence $\exists \epsilon'>0$ such that when $\sigma\leq \epsilon', |\gamma-\gamma^*| \leq \epsilon'$, we have 
\begin{align*}
h(\gamma)&\leq 2-H(N)-\eta(\epsilon_0)\\
h(\sigma) &\leq \epsilon_0\\
h(\gamma)+h(\sigma)-h(\sigma * \gamma)&\leq 2-H(N)-\eta(\epsilon_0)\\
h(\gamma)+h(\sigma) &\leq 2-H(N)-\eta(\epsilon_0)
\end{align*}

This implies that the source corresponding to $\sigma\leq \epsilon', \gamma\leq \gamma^*-\epsilon'$ can be transmitted reliably and the proof is complete.
\end{proof}

\input{Appendix.tex}

\input{references.tex}
\end{document}

%% file: MATH_template.tex
\usepackage{tikz}
\usepackage[utf8]{inputenc}

\usepackage{amsmath} \usepackage{amsthm} \usepackage{amsfonts} \usepackage{amssymb} 
\usepackage{epstopdf}
\usepackage{graphicx}
\input{xypic}
\usepackage{bbm}

\usepackage{enumerate}

\newtheorem{theorem}{Theorem}
\newtheorem{preposition}{Proposition}
\newtheorem{lem}{Lemma}

\newtheorem{definition}{Definition}

\theoremstyle{definition}
\newtheorem{example}{Example}

\newtheorem*{prob*}{Problem}

\theoremstyle{remark}
\newtheorem{remark}{Remark}

\global\long\def\ZZ{\mathbb{Z}}
\global\long\def\EE{\mathbb{E}}

\global\long\def\FF{\mathbb{F}}

%% file: Appendix.tex
\section{Proof of Theorem \ref{them: achievable-rate-for -proposed scheme}}\label{sec: proof of them lin_ces}
\begin{proof}
There are two error events, $E_0$ and $E_1$. $E_0$ occurs if no $\underline{\mathbf{\tilde{s}}}$ was found. $E_1$ is declared if  $\underline{\mathbf{\tilde{s}}} \neq \underline{\mathbf{s}}$. To show that $E_0$ is small, we need the next lemma. Suppose $v_i()$ and $v_i()$ are a realization of random functions generated as in the outline of the proof of Theorem \ref{them: achievable-rate-for -proposed scheme}.

\begin{lem}\label{len: prob jth comp}
Suppose $\mathbf{s}_i, i=1,2,3$ are jointly typical with respect to $P_{\underline{\mathbf{S}}}$. Then 
$$\big(v_1(\mathbf{s}_1),v_2(\mathbf{s}_2), v_3(\mathbf{s}_3),x_1(\mathbf{s}_1,v_1(\mathbf{s}_1)),x_2(\mathbf{s}_2,v_2(\mathbf{s}_2)),x_3(\mathbf{s}_3,v_3(\mathbf{s}_3))\big)\in A_{\epsilon}^{(n)}(V_1V_2V_3X_1X_2X_3|\mathbf{s}_1\mathbf{s}_2 \mathbf{s}_3).$$

%
%\begin{align*}
%P\{\mathbf{v}_i=\Phi_i(\mathbf{s}_i), \mathbf{x}_i=\mathbf{X}_i(\mathbf{s}_i, \mathbf{v}_i)  \mbox{for } l=1,2,3\}=\frac{1}{2^{2n}}\prod_{j=1}^n p(v_{3j}|v_{1j}, v_{2j},s_{1j}, s_{2j},s_{3j}) p(x_{1j}|s_{1,j}v_{1j})p(x_{1j}|s_{1,j}v_{1j})
%\end{align*}
\end{lem}
\begin{proof}
The proof is straightforward. 
\end{proof}
As a result, the sequences $\mathbf{s}_i, \mathbf{v}_i, \mathbf{x}_i, i=1,2,3$ are jointly typical with $\mathbf{y}$ with respect to $P_{\underline{\mathbf{S}}, \underline{\mathbf{V}}, \underline{\mathbf{X}},Y}$. This implies that  $P(E_0)$ approaches $0$ as $n \rightarrow \infty$. Next, we calculate $P(E_1 \cap E_0^c)$. For a given $\underline{\mathbf{s}}\in A_{\epsilon}(\underline{S})$, using the definition of $E_1$ and the union bound we obtain,
\begin{align*}
P(E_1 \cap E_0^c|\underline{\mathbf{s}}) \leq &  \sum_{ (\underline{\mathbf{v}},\underline{\mathbf{x}})\in  A_{\epsilon}(\underline{\mathbf{V}},\underline{\mathbf{X}}|\underline{\mathbf{s}})} \mathbbm{1}\{\mathbf{v}_i=v_i(\mathbf{s}_i), \mathbf{x}_i=x_i(\mathbf{s}_i, \mathbf{v}_i), ~ i=1,2,3\} \sum_{\mathbf{y}\in A_{\epsilon}(Y|\underline{\mathbf{x}})}p(\mathbf{y}|\underline{\mathbf{x}})\\
 &\sum_{\substack{(\underline{\mathbf{\tilde{s}}},\underline{\tilde{\mathbf{v}}},\underline{\tilde{\mathbf{x}}})\in A_{\epsilon}(\underline{S},\underline{V},\underline{X}|\mathbf{y})\\ \underline{\tilde{\mathbf{s}}} \neq \underline{\mathbf{s}}}} \mathbbm{1}\{\tilde{\mathbf{v}}_j=v_j(\tilde{\mathbf{s}}_j), \tilde{\mathbf{x}}_j=x_j(\tilde{\mathbf{s}}_j, \tilde{\mathbf{v}}_j), j=1,2,3\}
\end{align*}

Taking expectation over random functions  $X_i(,)$ and $V_i()$ gives,
\begin{align}\label{eq: Pe}
p_e(\underline{\mathbf{s}})=\EE\{P(E_{1}|\underline{\mathbf{s}})\} \leq &  \sum_{(\underline{\mathbf{v}}, \underline{\mathbf{x}},\mathbf{y})\in A_{\epsilon}(\underline{V}, \underline{X},Y|\underline{\mathbf{s}})} p(\mathbf{y}|\underline{\mathbf{x}})
 \sum_{\substack{(\underline{\mathbf{\tilde{s}}},\underline{\tilde{\mathbf{v}}},\underline{\tilde{\mathbf{x}}})\in A_{\epsilon}(\underline{S},\underline{V},\underline{X}|\mathbf{y})\\ \underline{\tilde{\mathbf{s}}} \neq \underline{\mathbf{s}}}} \\ \nonumber
& P\{\mathbf{v}_l=V_l(\mathbf{s}_l), \mathbf{x}_l=X_l(\mathbf{s}_l, \mathbf{v}_l), \tilde{\mathbf{v}}_l=V_l(\tilde{\mathbf{s}}_l), \tilde{\mathbf{x}}_l=X_l(\tilde{\mathbf{s}}_l, \tilde{\mathbf{v}}_l) ~ \mbox{for } l=1,2,3\}
\end{align}
%We need the following Lemma to proceed.

Note that $V_i()$ and $\mathbf{X}_i(~,~)$ are generated independently. So the most inner term in (\ref{eq: Pe}) is simplified to 
\begin{align}\label{eq: pp 1}
P\{\mathbf{v}_j=V_j(\mathbf{s}_j), \tilde{\mathbf{v}}_j=V_j(\tilde{\mathbf{s}}_j)~ j=1,2\} P\{\mathbf{x}_l=X_l(\mathbf{s}_l, \mathbf{v}_l), \tilde{\mathbf{x}}_l=X_l(\tilde{\mathbf{s}}_l, \tilde{\mathbf{v}}_l) ~  l=1,2,3\}.
\end{align}
Note $j=3$ is redundant because, $\mathbf{v}_3 = \mathbf{v}_1\oplus_q \mathbf{v}_2$ and $\mathbf{\tilde{v}}_3 = \mathbf{\tilde{v}}_1\oplus_q \mathbf{\tilde{v}}_2$. By definition, $V_j(\mathbf{s}_j)=\mathbf{s}_j\mathbf{G}+\mathbf{B}_j, j=1,2$, where $B_1,B_2$ are uniform and  independent of $\mathbf{G}$. Then 
\begin{align}\label{eq: PP 2}
P\{\mathbf{v}_j=\Phi_j(\mathbf{s}_j), \tilde{\mathbf{v}}_j=\Phi_j(\tilde{\mathbf{s}}_j)~ j=1,2\} = \frac{1}{q^{2n}} P\{(\tilde{\mathbf{s}}_j-\mathbf{s}_j)\mathbf{G}=\tilde{\mathbf{v}}_j-\mathbf{v}_j, ~j=1,2\}
\end{align}
The following lemma determines the above term.
\begin{lem}\label{lem: prob of sG}
Suppose elements of $\mathbf{G}$ are generated randomly and uniformly from $\FF_q$. If $\mathbf{s}_1$ or $\mathbf{s}_2$ is nonzero, the following holds: 
\begin{align*}
P\{\mathbf{s}_j\mathbf{G}=\mathbf{v}_j, ~j=1,2\}= \left\{ \begin{array}{ll}
%\mathbbm{1}\{\mathbf{v}_j=\mathbf{0}, ~ l=1,2\}, & \mbox{if}~ \mathbf{s}_1 = \mathbf{0},  \mathbf{s}_2=\mathbf{0}.\\
q^{-n}\mathbbm{1}\{\mathbf{v}_j=\mathbf{0}\} , & \mbox{if}~ \mathbf{s}_j=\mathbf{0}\\
%q^{-n}\mathbbm{1}\{\mathbf{\tilde{v}}_1=\mathbf{v}_1\} , & \mbox{if}~ \mathbf{\tilde{s}}_1 = \mathbf{s}_1,  \mathbf{\tilde{s}}_2\neq  \mathbf{s}_2.\\
q^{-n}\mathbbm{1}\{\mathbf{v}_1 = \mathbf{v}_2\}, & \mbox{if}~ \mathbf{s}_1 \neq \mathbf{0},  \mathbf{s}_2 \neq \mathbf{0}, \mathbf{s}_1= \mathbf{s}_2.\\
q^{-2n}, & \mbox{if}~ otherwise.\\
\end{array} \right.
\end{align*}
\end{lem}
\begin{proof}[Outline of the proof]
We can write $\mathbf{s}_j\mathbf{G}=\sum_{i=1}^n \mathbf{s}_{ji}\mathbf{G}_i$, where $\mathbf{s}_{ji}$ is the $i$th component of $\mathbf{s}_j$ and $\mathbf{G}_i$ is the $i$th row of $\mathbf{G}$. Not that $\mathbf{G}_i$ are independent random variables with uniform distribution over $\FF_q^n$. Hence, if $\mathbf{s}_j\neq \mathbf{0}$, then $\mathbf{s}_j\mathbf{G}$ is uniform over $\FF_q^n$. If $\mathbf{s}_1\neq \mathbf{s}_2$, one can show that $\mathbf{s}_1\mathbf{G}$ is independent of $\mathbf{s}_2\mathbf{G}$. The proof follows by arguing that if a random variables $X$ is independent of $Y$ and is uniform over $\FF_q$, then $X\oplus_q Y$ is also uniform over $\FF_q$ and is independent of $Y$. 
\end{proof}
Finally, we are ready to characterize the conditions in which $p_e\rightarrow 0$. We divide the last summation in (\ref{eq: Pe}) into the following cases:

\paragraph*{\bf Case 1, $\mathbf{\tilde{s}}_1 \neq \mathbf{s}_1,  \mathbf{\tilde{s}}_2= \mathbf{s}_2$}
In this case, using Lemma \ref{lem: prob of sG}, (\ref{eq: PP 2}) equals to $q^{-3n}\mathbbm{1}\{\mathbf{\tilde{v}}_2=\mathbf{v}_2\}$. Therefore, (\ref{eq: Pe}) is simplified to 
\begin{align*}
p_{e_1} (\underline{\mathbf{s}}) := &  \sum_{(\underline{\mathbf{v}}, \underline{\mathbf{x}},\mathbf{y})\in A_{\epsilon}(\underline{V}, \underline{X},Y|\underline{\mathbf{s}})} p(\mathbf{y}|\underline{\mathbf{x}})
 \sum_{\substack{(\underline{\mathbf{\tilde{s}}},\underline{\tilde{\mathbf{v}}},\underline{\tilde{\mathbf{x}}})\in A_{\epsilon}(\underline{S},\underline{V},\underline{X}|\mathbf{y})\\ \underline{\tilde{\mathbf{s}}} \neq \underline{\mathbf{s}},  \mathbf{\tilde{s}}_2= \mathbf{s}_2,  \mathbf{\tilde{v}}_2= \mathbf{v}_2 }} 
q^{-3n} P\{\mathbf{x}_l=\mathbf{X}_l(\mathbf{s}_l, \mathbf{v}_l), \tilde{\mathbf{x}}_l=\mathbf{X}_l(\tilde{\mathbf{s}}_l, \tilde{\mathbf{v}}_l) ~  l=1,2,3\}.
\end{align*} 

Note that $\mathbf{X}_l({\mathbf{s}}_l, \mathbf{{v}}_l)$ is independent of $\mathbf{X}_k(\tilde{\mathbf{s}}_k, \mathbf{\tilde{v}}_k)$, if $l \neq k$ or $\mathbf{s}_l \neq \mathbf{\tilde{s}}_l$ or $\mathbf{v}_l \neq \mathbf{\tilde{v}}_l$. Moreover, $P\{\mathbf{x}_l=\mathbf{X}_l(\mathbf{s}_l, \mathbf{v}_l)\}\approx 2^{nH(X_l|S_lV_l))}$.
As $\mathbf{s}_2=\mathbf{\tilde{s}}_2$ and $\mathbf{v}_2=\mathbf{\tilde{v}}_2$, then $X_2(\tilde{\mathbf{s}}_2, \tilde{\mathbf{v}}_2)=X_2(\mathbf{s}_2, \mathbf{v}_2)$. Therefore, 
\begin{align*}
P\{\mathbf{x}_l=X_l(\mathbf{s}_l, \mathbf{v}_l), \tilde{\mathbf{x}}_l=X_l(\tilde{\mathbf{s}}_l, \tilde{\mathbf{v}}_l) ~  l=1,2,3\}=2^{-n[2H(X_1|S_1V_1)+H(X_2|S_2V_2)+ 2H(X_3|S_3V_3)]} \mathbbm{1}\{\mathbf{\tilde{x}}_2=\mathbf{x}_2\}.
\end{align*}

Hence, we have:
\begin{align*}
p_{e_1}(\underline{\mathbf{s}})&\approx 2^{nH(\underline{V}, \underline{X}|\underline{S})} 2^{nH(S_1,V_1,X_1, S_3, V_3, X_3| Y S_2 V_2 X_2)} \frac{1}{q^{3n}}2^{-n[2H(X_1|S_1V_1)+H(X_2|S_2V_2)+2 H(X_3|S_3V_3)]}. 
\end{align*}
Note that $H(\underline{V},\underline{X}| \underline{S})=2\log_2q+\sum_{i=1}^3H(X_i|S_i,V_i)$. Therefore, $p_{e_1}\rightarrow 0$, if 
\begin{align}\label{eq: initial_inequality_case 1}
H(S_1,V_1,X_1, S_3, V_3, X_3| Y S_2 V_2 X_2)\leq \log_2q+ H(X_1|S_1V_1)+H(X_3|S_3V_3)
\end{align}
The right-hand side in the above inequality equals to $H(X_1X_3V_1V_3|S_1S_2S_3 X_2V_2)$. We simplify the left-hand side. Observe that $$H(S_1,V_1,X_1, S_3, V_3, X_3| Y S_2 V_2 X_2)=H(V_1,X_1, V_3, X_3| Y S_2 V_2 X_2)+H(S_1|S_2 \underline{V}~\underline{X}),$$
where $Y$ is removed from the second term, because conditioned on $\underline{X}$, $Y$ is independent of $S_1$. Note that
\begin{align*}
H(S_1|S_2 \underline{V}~\underline{X})&=H(S_1|S_2X_2V_2)-I(S_1;X_1V_1X_3V_3|S_2V_2X_2)=H(S_1|S_2)-I(S_1;X_1V_1X_3V_3|S_2V_2X_2).
\end{align*}
Therefore, using the above argument the inequality in (\ref{eq: initial_inequality_case 1}) is simplified to
\begin{align*}
H(S_1|S_2)&\leq I(S_1;X_1V_1X_3V_3|S_2V_2X_2)-H(V_1,X_1, V_3, X_3| Y S_2 V_2 X_2)+H(X_1X_3V_1V_3|S_1S_2S_3 X_2V_2)\\
%&=I(S_1;X_1V_1X_3V_3|S_2V_2X_2)-H(V_1,X_1, V_3, X_3| Y S_2 V_2 X_2)+H(X_1V_1X_3V_3|S_1S_2X_2V_2)\\
&=I(X_1V_1X_3V_3;Y|S_2V_2X_2)=I(X_1X_3;Y|S_2V_2X_2.)
\end{align*}

\paragraph*{\bf Case 2, $\mathbf{\tilde{s}}_1 = \mathbf{s}_1,  \mathbf{\tilde{s}}_2 \neq  \mathbf{s}_2$}
A similar argument as in the first case gives $H(S_2|S_1) \leq I(X_2X_3;Y|S_1V_1X_1)$.

\paragraph*{\bf Case 3, $\mathbf{\tilde{s}}_1 \neq \mathbf{s}_1,  \mathbf{\tilde{s}}_2 \neq  \mathbf{s}_2, \mathbf{\tilde{s}}_1 \oplus_q \mathbf{\tilde{s}}_2=\mathbf{s}_1\oplus_q \mathbf{s}_2 $}

Using Lemma \ref{lem: prob of sG},
\begin{align*}
P\{\mathbf{v}_j=\Phi_j(\mathbf{s}_j), \tilde{\mathbf{v}}_j=\Phi_j(\tilde{\mathbf{s}}_j)~ j=1,2\} = q^{-3n} \mathbbm{1}\{ \mathbf{\tilde{v}}_1 \oplus_q \mathbf{\tilde{v}}_2=\mathbf{v}_1\oplus_q \mathbf{v}_2 \}
\end{align*}
Therefore, the above probability is nonzero only when $\mathbf{\tilde{v}}_3=\mathbf{v}_3$. Hence, as $\mathbf{s}_3=\mathbf{\tilde{s}}_3$, we get $X_3(\tilde{\mathbf{s}}_3, \tilde{\mathbf{v}}_3)=X_3(\mathbf{s}_3, \mathbf{v}_3)$. This implies that,
\begin{align*}
P\{\mathbf{x}_l=\mathbf{X}_l(\mathbf{s}_l, \mathbf{v}_l), \tilde{\mathbf{x}}_l=\mathbf{X}_l(\tilde{\mathbf{s}}_l, \tilde{\mathbf{v}}_l) ~  l=1,2,3\}=2^{-n[2H(X_1|S_1V_1)+2H(X_2|S_2V_2)+H(X_3|S_3V_3)]} \mathbbm{1}\{\mathbf{\tilde{x}}_3=\mathbf{x}_3\}.
\end{align*}

As a result,  (\ref{eq: Pe}), in this case, is simplified to :
\begin{align*}
p_{e_3}(\underline{\mathbf{s}}) \approx 2^{nH(S_1,V_1,X_1, S_2, V_2, X_2| Y S_3 V_3 X_3)} q^{-n} 2^{-n[H(X_1|S_1V_1)+H(X_2|S_2V_2)]}.
\end{align*} 

Therefore, $p_{e3}\rightarrow 0$, if $H(S_1,V_1,X_1, S_2, V_2, X_2| Y S_3 V_3 X_3)\leq  H(X_1,X_2,V_1,V_2|S_1S_2S_3V_3X_3)$.
Using a similar argument as in the first case, this inequality is equivalent to $ H(S_1S_2|S_3) \leq I(X_1, X_2; Y| S_3V_3X_3)$.

\paragraph*{\bf Case 4, $\mathbf{\tilde{s}}_i \neq \mathbf{s}_i,  i=1,2,3 $}
Observe that,
\begin{align*}
P\{\mathbf{v}_j=\Phi_j(\mathbf{s}_j), \tilde{\mathbf{v}}_j=\Phi_j(\tilde{\mathbf{s}}_j)~ j=1,2\} &= q^{-4n}\\
P\{\mathbf{x}_l=\mathbf{X}_l(\mathbf{s}_l, \mathbf{v}_l), \tilde{\mathbf{x}}_l=\mathbf{X}_l(\tilde{\mathbf{s}}_l, \tilde{\mathbf{v}}_l) ~  l=1,2,3\}&=2^{-2n\sum_{l=1}^3 H(X_l|S_lV_l)}.
\end{align*}
Therefore,  (\ref{eq: Pe}), in this case, is simplified to $p_{e_4}(\underline{\mathbf{s}}) \approx q^{-2n}  2^{nH(  \underline{S}, \underline{V}, \underline{X}| Y)}  2^{-n\sum_{l=1}^3 H(X_l|S_lV_l)}$. As a result, one can show that $P_{e4}\rightarrow 0$, if $H(S_1S_2S_3) \leq I(X_1X_2X_3; Y)$.

Finally, note that $P_e(\underline{\mathbf{s}}) \leq \sum_{i=1}^4 P_{ei}(\underline{\mathbf{s}}).$ Moreover, $P_{ei}(\underline{\mathbf{s}})$ depends on $\underline{\mathbf{s}}$ only through its PMF. Therefore, for any typical $\underline{\mathbf{s}}$,  $P_e$ approaches zero as $n\rightarrow \infty$, if the following bounds are satisfied:
\begin{align*}
H(S_1|S_2) &\leq I(X_1 X_3;Y| S_2 V_2 X_2)\\
H(S_2|S_1) &\leq I(X_2 X_3;Y| S_1 V_1 X_1)\\
H(S_1S_2| S_1 \oplus_q S_2) & \leq I(X_1 X_2;Y| S_1 \oplus_q S_2,  V_3 X_3)\\
H(S_1,S_2) &\leq I(X_1 X_2 X_3;Y).
\end{align*}

\end{proof}